\documentclass[smallextended]{svjour3}       

\smartqed  

\usepackage{amsmath}
\usepackage{amssymb}
\usepackage{braket}
\usepackage{csquotes}
\usepackage[inline]{enumitem}
\usepackage{graphicx}
\usepackage[utf8]{inputenc}
\usepackage{mathtools}
\usepackage[numbers,sort&compress]{natbib}
\usepackage{todonotes}
\usepackage{xargs}
\usepackage{xspace}
\usepackage{algpseudocode}
\usepackage{algorithm}
\usepackage{subfig}
\usepackage{xfrac}

\usepackage{hyperref}     
\hypersetup{
	colorlinks=true,
	linkcolor=blue,
	citecolor=blue,
	filecolor=blue,
	urlcolor=blue,
}

\newcommandx*{\LDAUOmicron}[2][1=@pkling_false]{O\!\ifthenelse{\equal{#1}{small}}{\bigl(#2\bigr)}{\left(#2\right)}}
\newcommandx*{\LDAUomicron}[2][1=@pkling_false]{\mathrm{o}\!\ifthenelse{\equal{#1}{small}}{\bigl(#2\bigr)}{\left(#2\right)}}
\newcommandx*{\LDAUOmega}[2][1=@pkling_false]{\Omega\!\ifthenelse{\equal{#1}{small}}{\bigl(#2\bigr)}{\left(#2\right)}}
\newcommandx*{\LDAUomega}[2][1=@pkling_false]{\omega\!\ifthenelse{\equal{#1}{small}}{\bigl(#2\bigr)}{\left(#2\right)}}
\newcommandx*{\LDAUTheta}[2][1=@pkling_false]{\Theta\!\ifthenelse{\equal{#1}{small}}{\bigl(#2\bigr)}{\left(#2\right)}}

\setlist[description]{font=\normalfont\bfseries}
\setlist[enumerate,1]{label=(\alph*)}
\setlist[enumerate,2]{label=(\roman*)}

\begin{document}

\title{Approximation and Heuristic Algorithms for Computing Backbones in Asymmetric Ad-Hoc Networks\thanks{Based on ideas suggested and partially developed in two conference papers \cite{10fromalgo} and \cite{algo}. This work was partially supported by the German Research Foundation (DFG) within the Collaborative Research Center `On-The-Fly Computing' (SFB 901) and the International Graduate School `Dynamic Intelligent Systems'.}}

\author{Faisal N. Abu-Khzam \and Christine Markarian \and Friedhelm~Meyer~auf~der~Heide \and \\ Michael Schubert}

\institute{C.~Markarian and F.~Meyer~auf~der~Heide \at 
              Heinz Nixdorf Institute \& Computer Science Department,
              University of Paderborn,
              33102 Paderborn,
              Germany,
              \email{\{chrissm,fmadh\}@mail.upb.de}
           \and
           F. Abu-Khzam \at
              Department of Computer Science \& Mathematics, Lebanese American University, Beirut, Lebanon 
              \email{faisal.abukhzam@lau.edu.lb}
            \and
            M. Schubert \at 
            Department of Mathematics, 
            University of Paderborn,
            33102 Paderborn,
              Germany,
              \email{mischub@mail.upb.de}            
}

\date{Received: date / Accepted: date}

\maketitle
\begin{abstract}

We consider the problem of dominating set-based virtual backbone used for routing in asymmetric wireless ad-hoc networks. These networks have non-uniform transmission ranges and are modeled using 
the well-established disk graphs. The corresponding graph theoretic problem seeks a strongly connected dominating-absorbent set of minimum cardinality in a digraph. A subset of nodes in a digraph is a strongly connected dominating-absorbent set if the subgraph induced by these nodes is strongly connected and each node in the graph is either in the set or has both an in-neighbor and an out-neighbor in it. 

Distributed algorithms for this problem 
are of practical significance due to the dynamic nature of ad-hoc networks. 
We present a first distributed approximation algorithm, with a constant approximation factor and $O(Diam)$ running time, where $Diam$ is the diameter of the graph. Moreover we present a simple heuristic algorithm and conduct an extensive simulation study showing that our heuristic outperforms previously known approaches for the problem.
\end{abstract} 


\section{Introduction}\label{sec:introduction} Ad-hoc networks are local area networks built spontaneously as devices connect. Rather than relying on base stations to coordinate the flow of messages between nodes of the network, individual network nodes forward packets to and from each other without the use of pre–existing network infrastructure. Due to their self-dependent characteristic, ad-hoc networks have been largely deployed in many applications that preclude physical access such as disaster recovery and environmental monitoring. The main reason behind their efficiency has been the reliance on a virtual backbone formed by a subset of nodes in the network and acting as an underlying infrastructure. One of the most significant successes of virtual backbones has been in routing. Virtual backbones efficiently narrow down the search space of a route to the nodes in the backbone such that routing tables are maintained only by those nodes, which significantly reduces message overhead associated with routing updates (\cite{7fromthesis,15fromthesis,16fromthesis,17fromthesis,46fromthesis,48fromthesis}). 

\subsection{Dominating Set-based Virtual Backbones} To allow efficient routing, a virtual backbone is expected to be connected, as small as possible, and one hop away from all nodes of the network. If the devices forming a network use omnidirectional antennas, then the network can be modeled as a disk graph such that nodes represent the devices and disks represent the transmission ranges of nodes. An edge from node $u$ to node $v$ is added if $v$ lies in the disk of $u$. A wireless ad-hoc network is symmetric if all of its nodes have the same transmission range and asymmetric otherwise. A symmetric network is modeled as a Unit Disk Graph (UDG) and the underlying virtual backbone is a connected dominating set ($CDS$). A $CDS$ is a subset of nodes in an undirected graph such that each node in the graph is either in the subset or has a neighbor in it and the subgraph induced by the subset is connected. An asymmetric network is modeled as a Disk Graph (DG) and the underlying virtual backbone is a Strongly Connected Dominating Absorbent Set ($SCDAS$). An $SCDAS$ is a subset of nodes in a directed graph such that each node in the graph is either in the subset or has both an in-neighbor and an out-neighbor in it and the subgraph induced by the subset is strongly connected. 

In practice, nodes in a network often differ in power, control, or functionality and thus do not necessarily have the same transmission range. For example, nodes in power control schemes adjust their transmission power to save energy and reduce collisions. This results in unidirectional links between devices and hence asymmetry in the network. Therefore, in this paper, we will mainly focus on asymmetric wireless ad-hoc networks in which nodes have different transmission ranges.

\subsection{Algorithmic Challenges} Algorithms constructing a virtual backbone are often faced with a number of challenges. The most important ones are due to the dynamic nature of a wireless ad-hoc network. Nodes in such a network constantly change: some leave the network, new ones are added, and others change location. Hence, an algorithm constructing a virtual backbone should be able to adapt to these changes without losing its functionality. To this end, it must avoid any centralized computation where nodes are required to be aware of the entire network. This is especially essential in large networks in which by the time nodes gather information about the entire network, changes might have already been occurred. Thus, an algorithm for a virtual backbone is expected to be \emph{distributed} or \emph{local}. In a local algorithm, each node is able to identify whether it belongs to the backbone or not based on information it gathers only from nodes constant number of hops away from it. Moreover, when a failure occurs in some part of the network, only nodes in the vicinity of the failure get involved and locally fix the failure without affecting the whole network.

In addition to being distributed, an algorithm for a virtual backbone must use as few resources as possible while at the same time producing a good quality solution (i.e., a small $SCDAS$). One of the most well-studied of these resources is \emph{time}. 

In light of the above, we are interested in algorithms that are \emph{distributed}, \emph{fast}, and output a \emph{good quality solution}. The best possible solution would clearly be a smallest virtual backbone (i.e., smallest $SCDAS$). Finding a smallest $SCDAS$ (and $CDS$), however, is NP-hard even in UDGs \cite{1fromalgo}, which explains why heuristic methods have been mainly used in the literature. 

\subsection{Our contribution} We propose two distributed algorithms for the $SCDAS$ problem in Disk Graphs. The first is an approximation algorithm with a constant approximation factor and  an $O(Diam)$ running time, where $Diam$ is the diameter of the input graph. As of the writing of this paper, this is the first distributed algorithm with an approximation guarantee in DGs. When applied to Disk graphs with bidirectional edges (DGBs), our algorithm yields a constant approximation factor with $O(log^*n)$ running time that is optimal following the $\Omega(log^*n)$ lower bound by Lenzen et al. \cite{18fromalgo}. The second is a heuristic that outperforms all existing approaches in terms of $SCDAS$ size for DGs.

\section{Preliminaries}


Throughout this paper, basic graph theoretic notation such as \emph{degree} of a node, \emph{maximum degree} $\Delta$ in a graph, and diameter \emph{Diam} of a graph, are adopted. In a \emph{disk graph} $G = (V, E)$, each node $v \in V$ is fixed on the Euclidean plane and has a transmission range $r_v \in \left[r_{min},r_{max}\right]$, where $r_{min}$ and $r_{max}$ denote the minimum and maximum transmission range, respectively. For two disjoint nodes $u,v \in V$, there is a directed edge $(u,v)$ if and only if $d_{u,v} \leq r_u$ where $d_{u,v}$ is the \emph{Euclidean distance} between $u$ and $v$. An edge $(u,v) \in E$ is \emph{unidirectional} if $(v,u) \notin E$ and \emph{bidirectional} if $(v,u) \in E$. We say $G$ is (strongly) connected if for any two nodes $u,v \in V$, there exists a (directed) path from $u$ to $v$. An \emph{independent set} ($IS$) in $G$ is a subset $S$ of $V$ such that there is no bidirectional edge between any two nodes of $S$: we say $S$ does not violate independence. A subset of $V$ is a \emph{maximal independent set} (MIS) if it is an independent set to which no node can be added without violating independence. Let $(u,v)$ be a unidirectional edge. Then, $u$ is \emph{absorbed} by $v$ and $v$ is \emph{dominated} by $u$. A \emph{dominating set} is a subset of $V$ that dominates every node $v \in V$ and an \emph{absorbent set} is a subset of $V$ that absorbs every node $v \in V$.  An $r$-neighborhood of a node $v$ is the set of nodes in the graph that are within $r$ hops of $v$ (not including $v$ itself). We say a graph is \emph{growth-bounded} if there is a polynomial function $f(r)$ such that every $r$-neighborhood in the graph contains at most $f(r)$ independent nodes. 

In this paper, we adopt the model in which communication among nodes of a graph is done in synchronous rounds such that in each round, each node sends a message of size $O(\log n)$ bits to its neighbors. 

Initial work on dominating set-based virtual backbones focused on undirected general graphs and UDGs. There has been a lot of centralized \cite{24fromthesis,37fromthesis,6fromthesis} as well as distributed \cite{16fromthesis,17fromthesis,1fromthesis,2fromthesis,14fromthesis,22fromthesis,25fromthesis,50fromthesis,15fromthesis,49fromthesis,28fromthesis} approaches to construct a $CDS$ in the literature. The results were then extended to DGBs, subgraphs of DGs consisting of only bidirectional edges \cite{19fromalgo,21fromalgo,22fromalgo}.

DGs were first studied by Wu \cite{9frompaper}, who introduced the $SCDAS$ problem and gave a simple distributed algorithm with no approximation guarantee. In \cite{13fromalgo}, Clark et al. gave the first approximation algorithm, having a constant approximation factor, and two other heuristics outperforming the heuristic in \cite{9frompaper}. The two heuristics are (i) Dominating Absorbent Spanning Trees (DAST) and (ii) Greedy Strongly Connected Component Merging Algorithm (G-CMA). DAST constructs two spanning trees and outputs the union of the two trees as an $SCDAS$. G-CMA first finds a dominating absorbent set and then uses additional nodes to make the set strongly connected, using shortest paths between strongly connected components. A distributed heuristic for DGs by Kassaei et al. \cite{25fromalgo} was later shown to outperform the heuristics in \cite{13fromalgo} in terms of the $SCDAS$ size through simulations. 

As for general directed graphs, the only work done is by Li et al. \cite{14fromalgo} who proposed a centralized algorithm for $SCDAS$ with a logarithmic approximation factor, which is the best possible for general graphs \cite{4fromthesis} unless $P = NP$.

\section{An Approximation Algorithm}

In this section, we propose an approximation algorithm for $SCDAS$ in Disk Graphs and prove its constant approximation factor in $O(Diam)$ rounds. We also show that our algorithm results in constant approximation for $CDS$ in DGBs in $O(log^*n)$ rounds.

\subsection{The Algorithm}

Given a strongly connected directed graph $G = (V, E)$, our algorithm deletes all unidirectional edges from $G$ and then outputs an MIS $I$ (Figure \ref{fig:Alg1(i)}). Note that since $G$ is left with only bidirectional edges, we may apply any distributed algorithm for computing an MIS in undirected graphs, e.g. Luby's algorithm in \cite{Luby}. We will later present a faster algorithm to compute an MIS in DGs. Clearly, $I$ forms a dominating-absorbent set DAS in $G = (V, E)$. To strongly connect $I$ in $G$, the algorithm constructs $G' = (I, E')$ from $G$ such that a directed edge from $u$ to $v$ is added if there is a path from $u$ to $v$ of length at most three whose (at most two) inner nodes are in $V \setminus I$. This construction may lead to multiple edges from $u$ to $v$. In such a case, we remove all but one of them (without loss of generality). The nodes in $I$ along with the inner nodes in $G$ corresponding to the remaining edges of $G'$ form an $SCDAS$ (Figure \ref{fig:Alg1(ii)}).

\begin{figure} 
\begin{center}
    \includegraphics[scale=0.4]{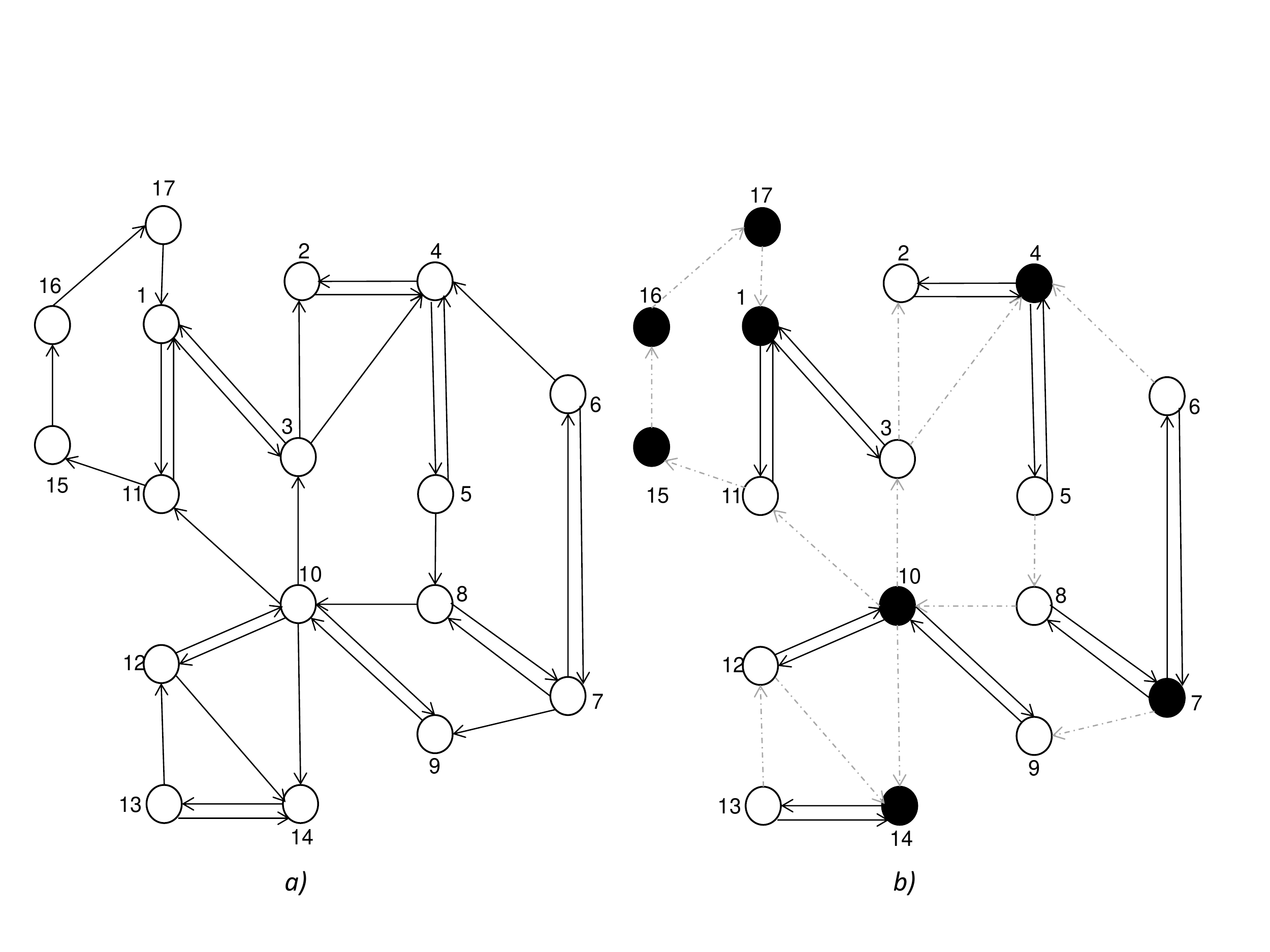}
		\caption{a) The original graph $G = (V, E)$, b) The dotted edges are the unidirectional edges that were removed from $G$ and the black nodes form the MIS nodes.}
		\label{fig:Alg1(i)}
\end{center}
\end{figure}
\begin{figure} 
 \begin{center}
    \includegraphics[scale=0.4]{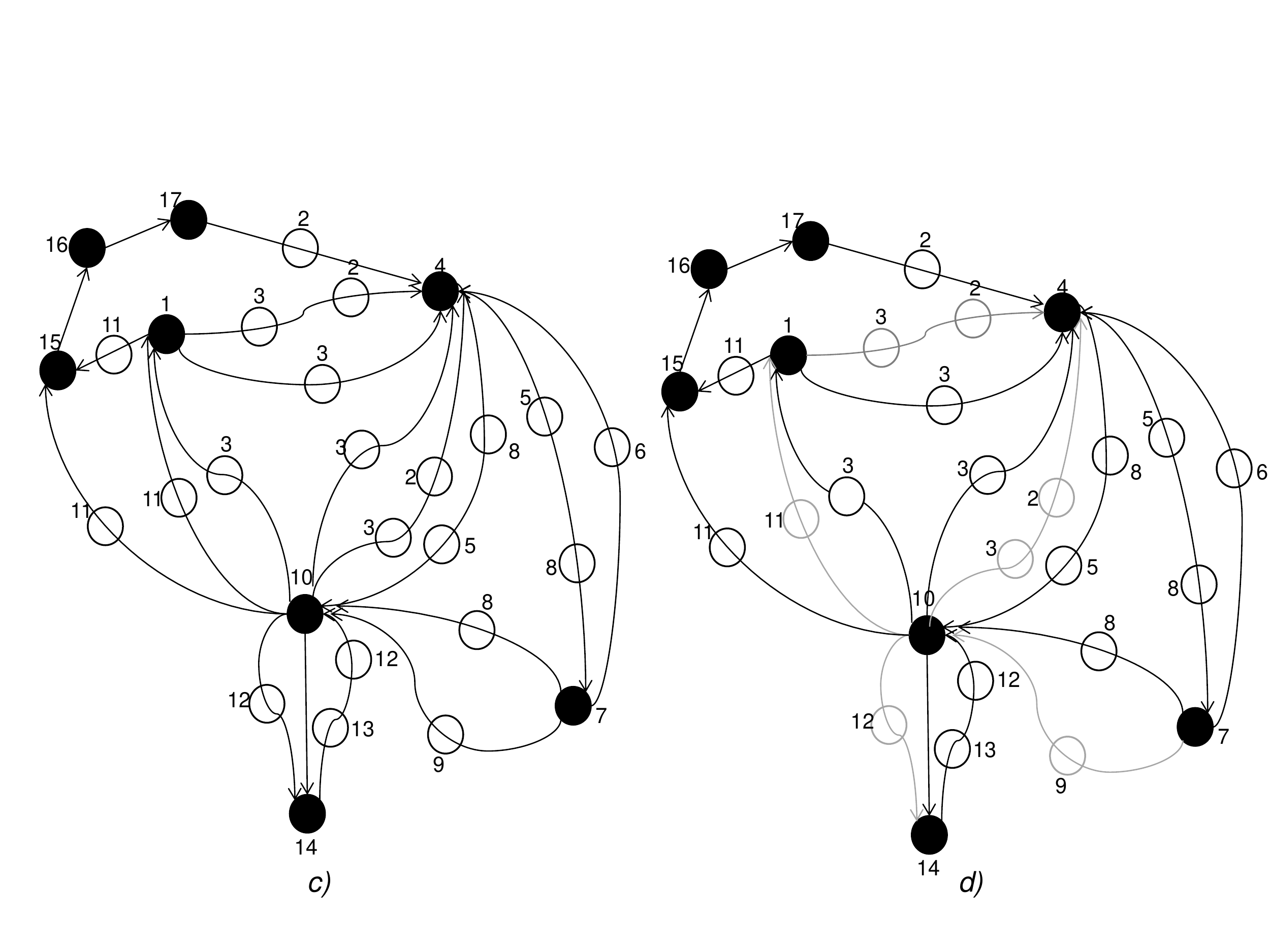}
		\caption{c) $G' =  (I,E')$ where the node set $I$ is in black. d) $G'=(I, E')$, the gray edges are those that were removed in Step 2.2, and the nodes on the black edges form the set $C$.}
		\label{fig:Alg1(ii)}
	\end{center}
\end{figure}

\begin{algorithm} 
  \caption{(Approximation Algorithm)}
  \begin{algorithmic}
	\State {\bf Input:} A strongly connected directed graph $G = (V,E)$
  \State {\bf Output:} An $SCDAS$ $S \subset V$
	\State Step 1: Delete all unidirectional edges from $G$ and then find an MIS $I$. 
  \State Step 2: Construct $G' = (I,E')$ from $G = (V,E)$ as follows. 
	\State \hspace{18pt} 2.1: For each pair $u$, $v$  of nodes from $I$, $E'$ contains a directed edge from $u$ to $v$ if there is a path from $u$ to $v$ of length at most three whose (at most two) inner nodes are in $V \setminus I$. This construction may lead to multiple edges from $u$ to $v$. 
  \State \hspace{18pt} 2.2: For each  pair  ($u, v$)  of nodes from $I$ such that $G'$ contains multiple edges from $u$ to $v$, remove all but one of these edges. 
	\State Output $I\cup C$, where $C$ is the set of the inner nodes in $G$ corresponding to the remaining edges of $G'$.
	\end{algorithmic}
	\label{alg:Approx}
\end{algorithm}
\paragraph{\bf Distributed Implementation:} The algorithm above is implemented as follows. All nodes are assigned IDs. (Step 1) In each round, each node runs a distributed MIS algorithm and decides whether it belongs to an MIS $I$ or not. (2.1) Each node in $I$ sends an \emph{edgeRequest} packet, a packet requesting to form an edge. The edgeRequest packet includes a \emph{source} that contains the ID of the node requesting to form an edge. If a non-MIS node in $V \setminus I$ receives such a request, it adds its ID to the packet and forwards the packet. Note that a node in $V \setminus I$ ignores all future edgeRequest packets from a source $u$ if it has previously received a packet from $u$. (2.2) A node in $I$ may receive multiple edgeRequests with the same source. For each source, it ignores all but the first edgeRequest. Thus with each source, it forms an edge consisting of the nodes with the IDs on the packet. Each node in $I$ informs the source nodes about its edge formations with them that in turn inform the participating nodes, i.e., nodes on the selected packets (edges). 

\subsection{Analysis}

In what follows, we show that the algorithm above is correct and yields a constant approximation factor in $O(Diam)$ rounds. 

\subsubsection{Correctness.} Given a strongly connected directed graph $G = (V, E)$, Algorithm \ref{alg:Approx} outputs an $SCDAS$.
\begin{proof}
The algorithm outputs $I\cup C$, where $I$ is a maximal independent set and therefore a dominating absorbent set. Thus, it remains to show that $I\cup C$ is strongly connected. In what follows, we show that if $G$ is strongly connected, then $G'$ is strongly connected. Let $G$ be a strongly connected graph and $I$ an MIS in $G$. Then, the following two properties hold.
\begin{enumerate}
	\item[i)] If there is a path $P = (u,v,w)$ with $u,v,w \in V \setminus I$ then $v$ must be dominated and absorbed by a node $y \in I$. This follows directly from the definition of a maximal independent set. 
	\item[ii)] For each pair of nodes $(u,w)$ there exists a walk $W = (u = v_1,...v_s = w)$ ($v_i$ not necessarily distinct) such that each node in $I \cap W$ is followed by at most two nodes of $V \setminus I$. Assume for contradiction that this does not hold. Let $A$ be the set including all walks $W$ from $u$ to $w$ in $G$. For a walk $W \in A$ let $m_W$ be the maximum length of a subpath given by consecutive nodes of $V \setminus I$ and $c(m_W)$ be the number of those subpaths with length $m_W$ in $W$. We choose a walk $W' = (u = v_1,...v_s = w) \in A$ such that $m_{W'}$ and $c(m_{W'})$ is minimum. Let $v_i,v_{i+1},v_{i+2}$ ($i=1..s-2$) be three nodes of a longest subpath with nodes in $V \setminus I$. Due to i, $v_{i+1}$ is dominated and absorbed by a node $y \in I$. Thus, $W = (v_1,...,v_i,v_{i+1},y,v_{i+1},v_{i+2},...v_s)$ is a walk from $u$ to $w$. If $c(m_{W'}) = 1$ then $m_{W} = m_{W'}-1$ and if $c(m_{W'}) > 1$ then $c(m_W) = c(m_{W'})-1$, contradicting the choice of $W'$.
\end{enumerate}
Let $I$ be the underlying MIS of $G$ and $u,w \in I$. It follows from ii) that we can find a walk $W$ such that each node in $I \cap W$ is followed by at most two nodes of $V \setminus I$. Thus, $W$ corresponds to a path $P$ in $G'$ from $u$ to $w$.
This implies that $G'$ constructed in Step 2.1 is strongly connected. Since Step 2.2 only removes edges from $E'$ whose deletion does not affect the strong connectivity of $G'$, the resulting  $I\cup C$ is strongly connected in $G$. \hfill $\square$
\end{proof}

\subsubsection{Approximation Ratio \& Running Time.} The theorem below shows that Algorithm \ref{alg:Approx} yields a constant approximation factor in $O(Diam)$ rounds.
 
\begin{theorem} 
Given a DG $G = (V, E)$ with transmission ratio $k = \frac{r_{max}}{r_{min}}$. Algorithm \ref{alg:Approx} gives $O(k^4)$-approximation factor for $SCDAS$ in $O(Diam)$ rounds.
\end{theorem}
\begin{proof} We show that at the end of Step 2.2, the degree of each node in $G'$ is upper bounded by $\left\lfloor 49k^2-1 \right\rfloor$. To prove the latter, observe that at the end of Step 2.2, each node in $G'$ has at most one out-going edge for each neighbor. Therefore, the degree of each node in $G'$ is bounded by the number of independent nodes in $u$'s 3-hop neighborhood in $G$. Since the distance between any two independent nodes $u,v$ in $G'$ is greater than $r_{min}$ (otherwise $u$ and $v$ must be connected in both directions), $d_{u,v}$ is thus bounded by $r_{min}< d_{u,v} < 3r_{max}$. The maximum area that may be covered by the disks of the independent nodes in $u$'s 3-hop neighborhood in $G$ is given by the difference of the areas between two disks with radii $3.5r_{max}$ and $\frac{r_{min}}{2}$, respectively. Furthermore, the minimum area of a disk is $\pi ( \frac{r_{min}}{2})^2$. It follows that $u$ has at most $\left\lfloor \frac{(3,5r_{max})^2-(\frac{r_{min}}{2})^2 }{(\frac{r_{min}}{2})^2} \right\rfloor = \left\lfloor 49k^2-1 \right\rfloor$ neighbors in $G'$.

Moreover, we have that the size of any independent set in G is upper bounded by $2.4(k+\frac{1}{2})^2 \cdot \left|SCDAS_{opt}\right| + 3.7(k+\frac{1}{2})^2$, where $SCDAS_{opt}$ denotes an optimal SCDAS (proved in \cite{Cent}).

Now we can conclude the bound on the approximation factor as follows. The size of the $SCDAS$ constructed by the algorithm is bounded by
\begin{equation}
\left|SCDAS\right|\leq \left|IS\right|( 1 + 2\Delta(G')),
\end{equation}
because Step 2 of the algorithm adds for each node in $I$ at most $2\Delta(G')$ nodes from $V \setminus I$ in order to strongly connect $I$. Plugging in the bounds from lemmas 3 and 4 yields:
\begin{equation}
\left|SCDAS\right|\leq\left(2.4(k+\frac{1}{2})^2\cdot\left|SCDAS_{opt}\right|+3.7(k+\frac{1}{2})^2\right)\cdot\left(1+2 \left\lfloor49k^2-1\right\rfloor\right)
\end{equation}
This implies:
\begin{equation}
\left|SCDAS\right|= O(k^4) \cdot \left|SCDAS_{opt}\right| 
\end{equation}

It remains to show that Algorithm \ref{alg:Approx} takes $O(Diam)$ rounds. To compute an MIS in Step 1, we use Luby's $O(\log n)$ time randomized algorithm to construct an MIS for directed general graphs. As for DGs, there is a deterministic $O(\log^*n)$ time algorithm to find an MIS in bounded-growth graphs \cite{2008}. If the  transmission ratio  $k = r_{max}/r_{min}$ is bounded, DGs become bounded-growth graphs with $f(r) = O(r^2 k^2)$. Therefore, for DGs, rather than using Luby's $O(\log n)$ time randomized algorithm, we use the deterministic $O(\log^*n)$ time algorithm for bounded-growth graphs in \cite{2008} which when applied to DGs takes $O(k^8 \log^*n)$ time. In Step 2, constructing $G' = (I,E')$ from $G = (V,E)$ needs only three broadcasts because the edgeRequest packets stop after at most two inner nodes. Once `connecting' nodes are selected, it remains to inform them. Informing the participating nodes takes $O(Diam)$ rounds where each source node informs other source nodes about the selected edges and consequently the selected nodes. Note that the number of nodes to be informed are at most $2k^2$, thus bounding the number of propagated messages by $O(k^2)$.Therefore, the total running time is bounded by $O(Diam)$.  \hfill $\square$
\end{proof}
\subsection{Discussion} Here we show that a slight modification of Algorithm \ref{alg:Approx} yields a constant approximation for CDS in DGBs in $O(\log^*n)$ rounds. 

Given a connected undirected graph $G = (V, E)$, the algorithm constructs an MIS $I$ in $G$ that clearly forms a dominating set in $G$. To connect $I$ in $G$, each node in $I$ ignores all edgeRequest packets from all source nodes except the first one it receives. Since the graph is undirected, it is enough for each node in $I$ to connect to one other node in $I$. 

\begin{theorem} 
Given a DGB $G = (V, E)$ with transmission ratio $k = \frac{r_{max}}{r_{min}}$. 
Algorithm \ref{alg:Approx} gives $O(\ln k)$-approximation factor for $CDS$ in $O(\log^*n)$
rounds.
\end{theorem}
\begin{proof}
The size of any independent set in G is at most $O(\ln k)$ (proved in \cite{Wang}). Hence, since each node in $I$ adds at most two nodes from $V \setminus I$ to connect, the approximation factor holds. 

To compute an MIS in $G$, we use the same deterministic $O(\log^*n)$ time algorithm to find an MIS in bounded-growth graphs \cite{2008} which when applied to DGBs, takes $O(k^8 \log^*n)$ time. Moreover, since the graph is undirected, $O(1)$ time is needed to connect the constructed MIS, thus completing the proof. \hfill $\square$
\end{proof}

\section{A Simple Heuristic Approach}

In this section, we propose a heuristic for $SCDAS$ in DGs and show through extensive simulations that it outperforms all existing approaches for DGs in terms of $SCDAS$ size.

\subsection{The Algorithm}

We first give the intuitive ideas behind our algorithm and then describe it formally. 

A high-degree node is more likely to dominate and absorb other nodes than nodes of lower degree. Thus, a `good' solution will most likely contain many high-degree nodes. However, since the solution set must also be strongly connected, some low-degree nodes might also be used as intermediary nodes to connect the nodes in an optimal solution set. Instead of discovering the low-degree nodes in that are needed only to guarantee strong connectivity, our algorithm starts by deleting a low-degree node that does not contribute to the strong connectivity property. This can easily be done by removing a low-degree node from the graph and checking if the graph remains strongly connected. If so, the just-removed node can be deleted (i.e., not counted in the solution). When such a node is deleted, it has to be dominated and absorbed by a subset of the remaining nodes. Therefore, high-degree in- and out-neighbors are selected (i.e., added to the solution) to dominate and absorb the deleted node. Once a node is selected, it can never be deleted. 

More formally, we color the nodes of the graph in green, red, or white as follows. Initially, all nodes are white and form an $SCDAS$. A green node is one that is decided to be in the solution, while a red one is a deleted node (decided not to be in). A white node is not yet decided. At any subsequent stage, the white and green nodes induce an $SCDAS$.
The algorithm ends when the set of white nodes becomes empty. In addition to the above procedure, when the number of non-red in-neighbors (or out-neighbors) of a white node $v$ drops to one, the only one in-neighbor (or out-neighbor) of $v$ is automatically placed in the solution (if not already green), being the only node that can absorb (or, respectively, dominate) $v$.

\begin{algorithm} 
 \caption{(LDHD: Low-Degree Elimination and High-Degree Selection)}
 \begin{algorithmic}
	\State {\bf Input:} A strongly connected directed graph $G = (V,E)$
  \State {\bf Output:} An $SCDAS$ $S \subset V$
	\State Initially all nodes are in $S$ and colored white. While $S$ has white nodes, 
	\State Select a white node $v$ of minimum degree.
  \State \hspace{18pt} If $G[S \backslash v]$ is not strongly connected
  \State \hspace{18pt} Color $v$ green
  \State Else
  \State \hspace{18pt} Color $v$ red and update the degrees of its neighbors 
  \State \hspace{18pt} If none of $v$'s in-neighbors is green
  \State \hspace{36pt} Select in-neighbor $u$ of maximum degree 
  \State \hspace{36pt} Color $u$ green
  \State \hspace{18pt} If none of $v$'s out-neighbors is green
  \State \hspace{36pt} Select out-neighbor $w$ of maximum degree
  \State \hspace{36pt} Color $w$ green
	\State Update the degrees

\end{algorithmic}
	\label{alg:LDHD}
\end{algorithm}

\paragraph{\bf Distributed Implementation:} The algorithm LDHD above is implemented as follows. All nodes are assigned IDs. After all nodes acquire information about the degrees and IDs of their in- and out- neighbors, a node $v$ with minimum degree among its neighbors initiates the algorithm (ties are broken by smaller ID first). To check if removing $v$ disconnects $G$, it is enough to check if the subgraph induced by the neighbors (in- and out- neighbors) of $v$ is disconnected in which case $v$ goes into the solution. Otherwise, the in- and out- neighbors with highest degree (ties are again broken by smaller ID first) go into the solution to dominate and absorb $v$, respectively (if not already dominated and/or absorbed). Information about the degrees are broadcasted after being updated. 

\paragraph{\bf Running Time:} LDHD requires $O(Diam)$ rounds every time nodes gather information about their neighborhood. Each node waits for lower degree nodes at most $O(\Delta) = O(n)$ rounds until it is its turn. Therefore, the overall running time of the algorithm is at most $O(n \cdot Diam)$. 

\subsection{Experimental Analysis}

In this section, we present the results of our extensive simulations conducted to evaluate the performance of LDHD. 

LDHD is compared to the exact solution produced by a brute force optimization algorithm as well as to two well known approaches in the literature: Dominating Absorbent Spanning Trees (DAST) and Greedy Strongly Connected Component Merging algorithm (G-CMA) proposed by \cite{13fromalgo}. The quality measure is the size of the $SCDAS$ constructed. 

To generate a random asymmetric network, we follow the same approach adopted in previously published work (e.g., in \cite{13fromalgo}). $N$ nodes with distinct identity numbers between 1 and $N$ are located (randomly) in a limited square area of the Euclidean plane. Each node chooses a random transmission range that is bounded by some specified minimum and maximum values. A directed edge is added from a node $u$ to node $v$ if the Euclidean distance between $u$ and $v$ is less than the transmission range of $u$. If the generated graph is strongly connected, we use it as a test instance, otherwise we discard it. We measure the performance of each approach under the effect of two network parameters:

\begin{enumerate}
\item Network Density, which we vary in two ways:
\begin{enumerate}
\item Different numbers of nodes in a fixed area
\item Different area sizes for a fixed number of nodes
\end{enumerate}
\item Transmission ratio: $k = T_{rmax}/$ $T_{rmin}$, where 
$T_{rmax}$ and $T_{rmin}$ are the maximum and minimum transmission ranges respectively.
\end{enumerate}

\noindent
Simulation for each performance measure is repeated 100 times for each instance and the average result is taken.

\subsubsection{Network Density: Different Number of Nodes}

To compare the size of the solution constructed by each of the four algorithms, we deploy $N$ nodes in a 1000m x 1000m  area. $N$ changes between 10 and 130 with an increment of 10. The nodes select their transmission ranges from the interval [$T_{rmin}$ , $T_{rmax}$ ]. Figure \ref{fig1} shows the performance of each approach.

\subsubsection{Network Density: Different area size}

To study the effect of varying the area on the performance of each approach, we deploy a fixed number of nodes, $N=50$. The nodes select their transmission ranges from the interval 
[$T_{rmin}$ , $T_{rmax}$]. The area varies from 600m x 600m to 1400m x 1400m. Figure \ref{fig2} shows the performance of each approach.

\subsubsection{Transmission ratio}
We also study the effect of varying the transmission ratio $k = T_{rmax}/$ $T_{rmin}$ on the size of the solutions constructed by each of the four approaches. We conducted two experiments. In the first, we randomly locate 50 nodes in a fixed 1000m x 1000m area and let $k$ varies as follows: We fix $T_{rmax}$ = 1000m and vary $T_{rmin}$ between 200m and 1000m with an increment of 200 for 
$k = 1$ to $5$. In the second experiment, we measure the performances on a larger network and randomly locate 100 nodes in a fixed 1200m x 1200m area and vary $k$ as follows: we fix $T_{rmax}$ = 1200m and vary $T_{rmin}$ between 200m and 1200m with an increment of 200 for $k=1$ to $6$. Figure \ref{fig3} and Figure \ref{fig4} show the performance of each approach in each of the experiments, respectively. 

\begin{figure} 
\begin{center}
\includegraphics[scale= 0.35]{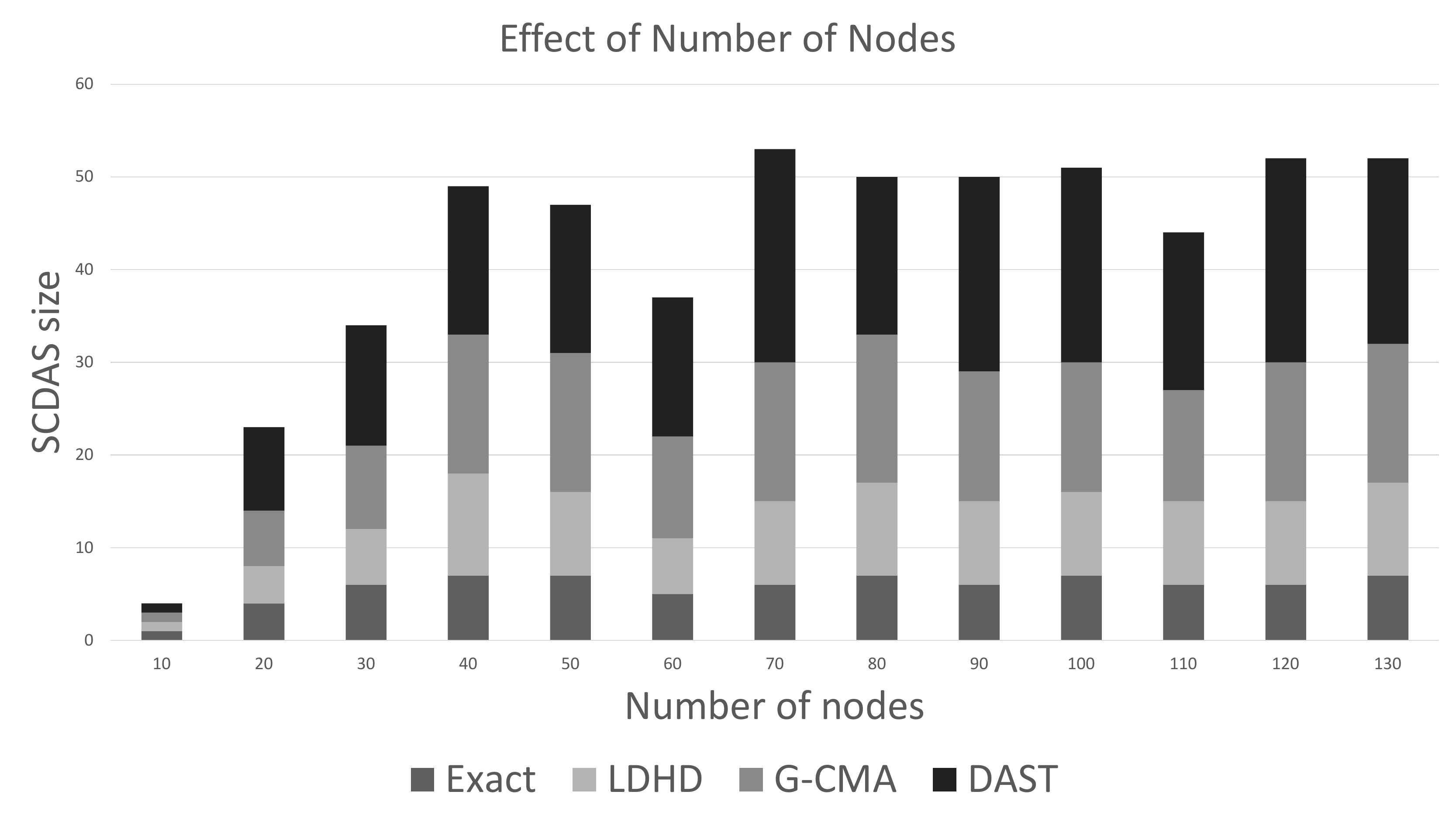}
\caption{Network Density: Different number of nodes}
\label{fig1}
\end{center}
\end{figure}

\begin{figure} 
\begin{center}
\includegraphics[scale= 0.35]{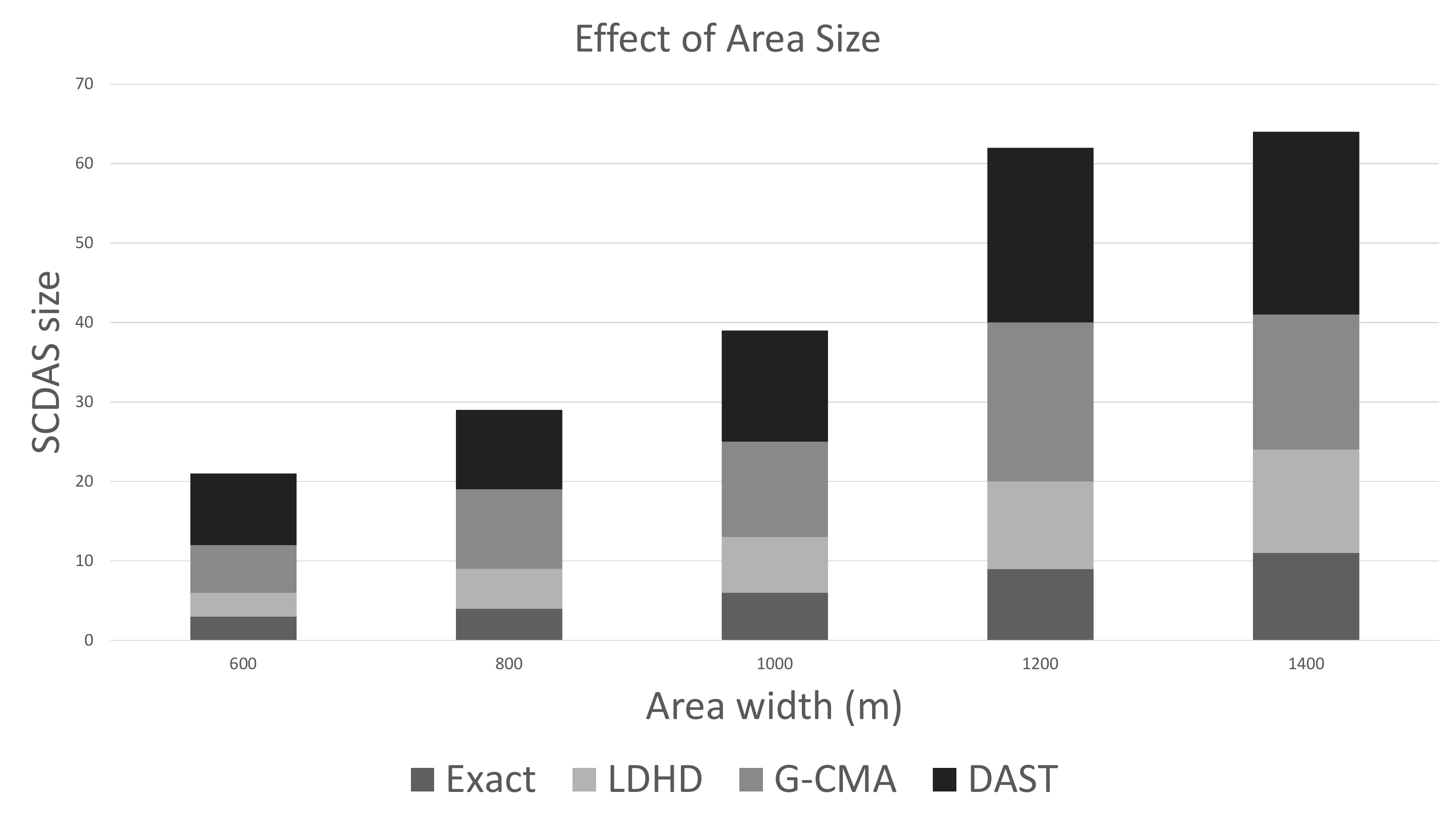}
\caption{Network Density: Different area size}
\label{fig2}
\end{center}
\end{figure}

\begin{figure} 
\begin{center}
\includegraphics[scale= 0.35]{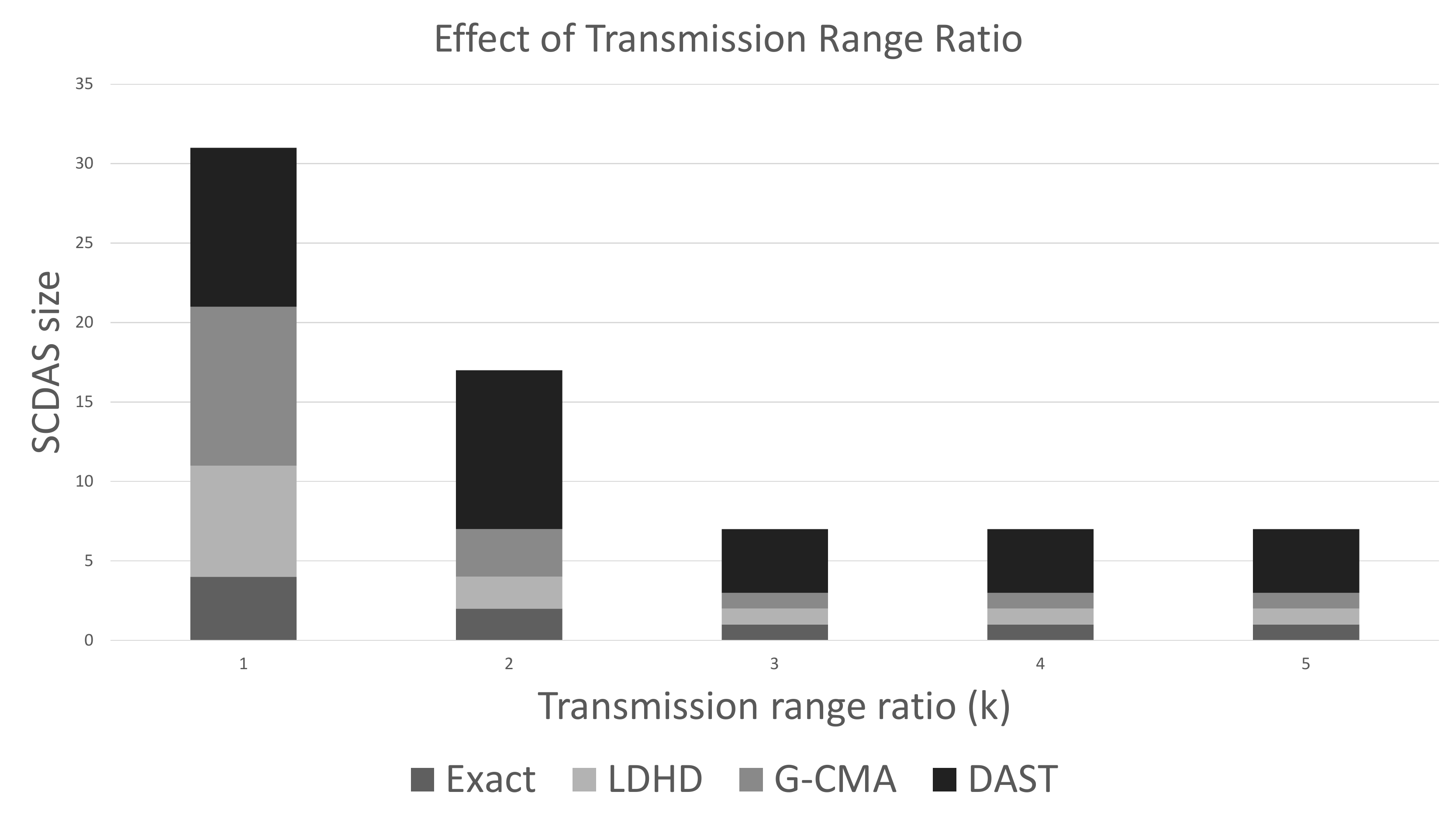}
\caption{Different Transmission Ratios; N=50}
\label{fig3}
\end{center}
\end{figure}

\begin{figure} 
\begin{center}
\includegraphics[scale= 0.35]{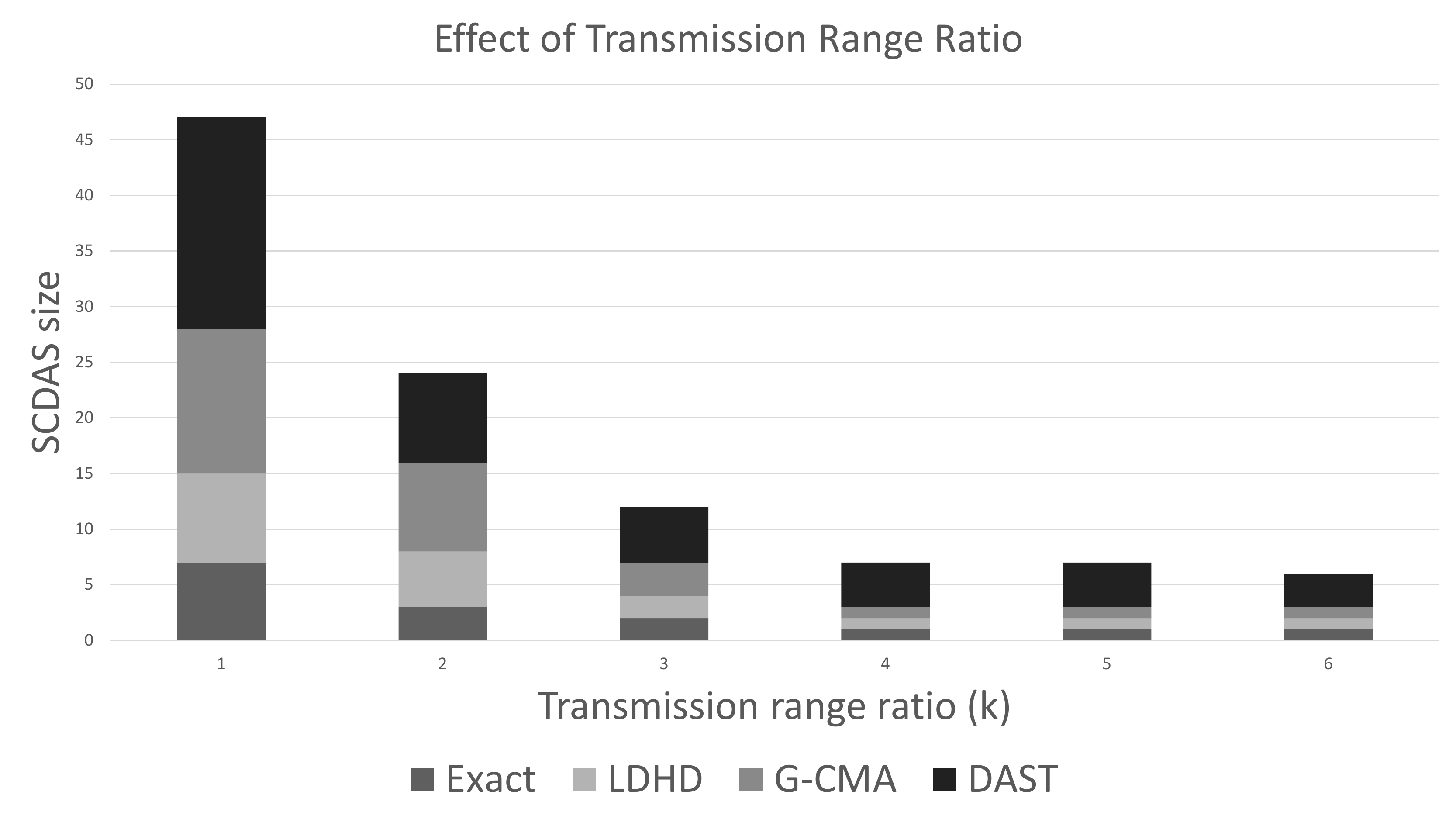}
\caption{Different Transmission Ratios; N=100}
\label{fig4}
\end{center}
\end{figure}

\subsection{Discussion}

As shown in the four charts, there is a notable difference in the size of the $SCDAS$ constructed by each of the four approaches. Obviously, LDHD is the closest to optimum in all the different experimental setups. In fact, the size of our computed solution was never larger than 1.75 times the optimum, while the sets constructed by the two other heuristics were in the range: 2.5 to 5 times the optimum.

The main strength of LDHD is its efficiency, simplicity as well as the quality of delivered solutions: despite being very simple, LDHD can deliver solutions whose size is much closer to optimum than those produced by other algorithms found in the literature. In fact, one cannot neglect the need for highly efficient simple algorithms considering the limited computational resources, especially in wireless networks. 

Another remarkable feature of LDHD is that it maintains a feasible solution at any point during the search for a best-possible solution. Therefore, a solution can be delivered whenever needed, should time be the most critical measure. Note that the algorithm proposed by \cite{25fromalgo} can be characterized with the same feature of maintaining a feasible solution at any point but involves more computation compared to LDHD. Our algorithm considers the role played by both low- and high- degree nodes, whereas the algorithm in \cite{25fromalgo} considers only low-degree nodes $v$ and includes two more tests in addition to \emph{strong connectivity test} (whether removing $v$ disconnects the graph): \emph{domination test} (whether all nodes dominated by $v$ can be dominated by another node) and \emph{absorbency test} (whether all nodes absorbed by $v$ can be absorbed by another node). A low-degree node $v$ is then added to the solution if it passes the three tests. An important observation here is that if $v$ does not disconnect the remaining graph (i.e., passes the strong connectivity test), this means all of $v$'s neighbors can reach and be reached by all other nodes and must therefore be absorbed and dominated by other nodes (i.e., this means $v$ automatically passes the two other tests too and there is no point of taking them). 

\section{Concluding Remarks}
 
This paper provides distributed algorithms for the problem of dominating set-based virtual backbone used for routing in asymmetric wireless ad-hoc networks. The techniques and ideas developed in this paper can be of independent interest and can be used in other extensions of the problem, thus capturing more challenges faced by these networks. For example, nodes in the virtual backbone are often subject to failure, thus leaving fault tolerance unavoidable. Symmetric networks have been extensively studied with the fault tolerant consideration, whereas only few heuristics have studied the problem in asymmetric networks \cite{conjecture,simulation}. Hence, it will be interesting to give approximation guarantees and better heuristics for the underlying problem in asymmetric networks.  

Another interesting direction is to explore the virtual backbone problem under the effect of unknown future. Previous work in the literature has mainly addressed the problem with the assumption that the entire network is known in advance. It will be interesting to model the problem in an online setting in which nodes are revealed with time and an online algorithm has to construct an efficient virtual backbone at each point of time without knowing future nodes. Within this context, there has been some work on the equivalent problem of dominating set in general graphs: the online set cover problem by Alon et al.\cite{onlinesetcover}. To the best of our knowledge, no one has considered the online dominating set problem in Disk Graphs. Nevertheless, similar problems such as finding large independent sets in Disk Graphs have been studied in an online setting \cite{onlineIndependentSet} and might give some insights to solve the online dominating set problem in Disk Graphs.



{} 

\begin{thebibliography}{1}


\bibitem{10fromalgo} F. Abu-Khzam and C. Markarian. {\em A Degree-Based Heuristic for Strongly Connected Dominating-Absorbent Sets in Wireless Ad-Hoc Networks. Innovations in Information Technology}, pp. 200-204, 2012.

\bibitem{algo} C. Markarian, F. Meyer auf der Heide, and M. Schubert. {\em A Distributed Approximation Algorithm for Strongly Connected Dominating-Absorbent Sets in Asymmetric Wireless Ad-Hoc Networks. ALGOSENSORS}, pp. 217-227, 2013.


\bibitem{7fromthesis} M. Cardei, X. Cheng, X. Cheng, and D. zhu Du. {\em Connected Domination in Multihop Ad-hoc Wireless Networks. International Conference on Computer Science and Informatics}, pp. 251-255, 2002.

\bibitem{46fromthesis} J. Wu. {\em Extended Dominating-set-based Routing in Ad-hoc Wireless Networks with Unidirectional Links. IEEE Transactions on Parallel and Distributed Systems}, 13 (9) pp. 866-881, 2004.

\bibitem{48fromthesis} J. Wu, F. Dai, M. Gao, and I. Stojmenovic. {\em On Calculating Power-aware Connected Dominating Sets for Efficient Routing in Ad-hoc Wireless Networks. IEEE Journal of Communications and Networks}, 4 (1) pp. 59-70, 2002. 

\bibitem{24fromthesis} S. Guha and S. Khuller. {\em Approximation Algorithms for Connected Dominating Sets. Algoritmica}, 20 (4) pp. 374-387, 1998.

\bibitem{37fromthesis} M. Min, H. Du, X. Jia, C. X. Huang, S. C. H. Huang, and W. Wu. {\em Improving Construction for Connected Dominating Set with Steiner Tree in Wireless Sensor Networks. Journal of Global Optimization}, 35 (1) pp. 111-119, 2006.

\bibitem{6fromthesis} S. Butenko, X. Cheng, C. A. Obviera, and P. Pardalos. {\em A New Heuristic for the Minimum Connected Dominating Set Problem on Ad-hoc Wireless Networks. Recent Developments in Cooperative Control and Optimization}, pp. 61-73, 2004.

\bibitem{4fromthesis} N. Alon, D. Moshkovitz, and S. Safra. S. {\em Algorithmic Construction of Sets for k-restrictions. ACM Transactions on Algorithms} 2 (2) pp. 153-177, 2006.

\bibitem{16fromthesis} B. Das and V. Bharghavan. {\em Routing in Ad-hoc Networks Using Minimum Connected Dominating Sets. International Conference on Communications}, pp 376-380, 1997.

\bibitem{17fromthesis} B. Das, R. Sivakumar, and V. Bharghavan. {\em Routing in Ad-hoc Networks Using a Spine. International Conference on Computers and Communication Networks}, pp 34-39, 1997.

 \bibitem{1fromthesis} K. Alzoubi, R-J. Wan, and O. Frieder. {\em New Distributed Algorithm for Connected Dominating Set in Wireless Ad-hoc Networks. International Conference on System Sciences}, pp. 3849-3855, 2002.

\bibitem{2fromthesis} K. M. Alzoubi, R-J. Wan, and O. Frieder. {\em Message-optimal Connected Dominating Sets in Mobile Ad-hoc Networks. ACM International Symposium on Mobile Ad-hoc Networking and Computing}, pp. 157-164, 2002.

\bibitem{14fromthesis} J. Czyzowicz, S. Dobrev, T. Fevens, H. González-Aguilar, E. Kranakis, J. Opatrny, and J. Urrutia. {\em Local Algorithms for Dominating and Connected Dominating Sets of Unit Disk Graphs with Location Aware Nodes. Latin American Symposium on Theoretical Informatics}, 4957 pp. 158-169, 2008.

\bibitem{22fromthesis}  S. Funke, A. Kesselman, U. Meyer, and M. Segal. {\em A Simple Improved Distributed Algorithm for Minimum CDS in Unit Disk Graphs. ACM Transactions on Sensor Networks}, 2 (3) pp. 444-453, 2006.
\bibitem{25fromthesis} B. Han. {\em  Zone-based Virtual Backbone Formation in Wireless Ad-hoc Networks. Ad-hoc Networks}, 7 (l) pp. 183-200, 2009.


\bibitem{50fromthesis} J. Wu and H. Li. {\em On Calculating Connected Dominating Set for Efficient Routing in Ad-hoc Wireless Networks. International Workshop on Discrete Algorithms and Methods for Mobile Computing and Communications}, pp. 7-14, 1999.
\bibitem{15fromthesis} F. Dai and J. Wu. {\em An Extended Localized Algorithm for Connected Dominating Set Formation in Ad-hoc Wireless Networks. IEEE Transactions on Parallel and Distributed Systems}, 15 (10) pp. 908-920, 2004.

\bibitem{49fromthesis} J. Wu, M. Gao, and I. Stojmenovic. {\em On Calculating Power-aware Connected Dominating Sets for Efficient Routing in Ad-hoc Wireless Networks. International Conference on Parallel Processing}, pp. 346-356, 2001.

\bibitem{28fromthesis} H. Kassaei, M. Mehrandish, L. Narayanan, and J. Opatrny. {\em A New Local Algorithm for Backbone Formation in Ad-hoc Networks. ACM Symposium on Performance Evaluation of Wireless Ad-hoc, Sensor, and Ubiquitous Networks}, pp. 49-57, 2009.

	\bibitem{19fromalgo} M. Thai, F. Wang, D. Liu, S. Zhu, and D. Du. {\em Connected Dominating Sets in Wireless Networks with Different Transmission Ranges. Mobile Computing}, 6 (7), 2007.
	
	
	\bibitem{21fromalgo} H. Raei, M. Fathi, A. Akhhlaghi, and B. Ahmadipoor. {\em A New Distributed Algorithm for Virtual Backbone in Wireless Sensor Networks with Different Transmission Ranges. Computer Systems and Applications}, pp. 983-988, 2009.
	
	\bibitem{22fromalgo} H. Raei, M. Sarram, B. Salimi, and F. Adibnya. {\em Energy-Aware Distributed Algorithm for Virtual Backbone in Wireless Sensor Networks. Innovations in Information Technology}, pp. 435-439, 2008. 
	
	  \bibitem{9frompaper} J. Wu. {\em An Extended Dominating-set-based Routing in Ad hoc Wireless Networks with Unidirectional links. IEEE Transactions on Parallel and Distributed Systems}, 13 (9) pp. 866-881, 2002.
		
	\bibitem{1fromalgo} B. Clark, C. Colbourn, and D. Johnson. {\em Unit Disk Graphs. Discrete Mathematics}, 86 pp. 165-177, 1990. 


  \bibitem{13fromalgo} M. A. Park, J. Willson, C. Wang, M. Thai, W. Wu, and A. Farago. {\em A Dominating and Absorbent Set in a Wireless Adhoc Network with Different Transmission Ranges. Mobile Ad hoc Networking and Computing}, pp. 22-31, 2007.
	
	  \bibitem{14fromalgo} D. Li, H. Duc, and P. Wan. {\em Construction of Strongly Connected Dominating Sets in Asymmetric Multihop Wireless Networks.  Theoretical Computer Science}, 410 pp. 661-669, 2009.

	\bibitem{18fromalgo} C. Lenzen and R. Wattenhofer. {\em Leveraging Linial’s locality limit. Distributed Computing}, 5218 pp. 394–407, 2008.

	\bibitem{Luby} M. Luby. {\em A Simple Parallel Algorithm for the Maximal Independent Set Problem. SIAM Journal on Computing}, 15 (4) pp. 1036-1053, 1986.
	
	\bibitem{Cent} M. A. Park, J. Willson, C. Wang, M. Thai, W. Wu, and A. Farago. {\em A Dominating and Absorbent Set in a Wireless Adhoc Network with Different Transmission Ranges. Mobile Ad hoc Networking and Computing}, pp. 22-31, 2007.
	
		\bibitem{Wang} M. Thai, F. Wang, D. Liu, S. Zhu, and D. Du. {\em Connected Dominating Sets in Wireless Networks with Different Transmission Ranges. Mobile Computing}, 6 (7), pp. 721-730, 2007.
		
  \bibitem{2008} J. Schneider, R. Wattenhofer. {\em A log-star Distributed Maximal Independent Set Algorithm for Growth-Bounded Graphs. Principles of Distributed Computing}, pp. 35-44, 2008.

\bibitem{25fromalgo} H. Kassaei and L. Narayanan. {\em A New Algorithm for Backbone Formation in Ad hoc Wireless Networks of Nodes with Different Transmission Ranges. Wireless and Mobile Computing}, pp. 83-90, 2010.


	\bibitem{16fromalgo} R. Tiwari, T. Mishra, and Y. Li. {\em  k-strongly Connected Dominating and Absorbing Set in Wireless Ad hoc Networks with Unidirectional Links. Wireless Algorithm, Systems and Applicaitons}, pp. 103-112, 2007.
	
	\bibitem{onlinesetcover} N. Alon, B. Awerbuch, Y. Azar, N. Buchbinder, and J. Naor. {\em The Online Set Cover Problem. ACM Symposium on the Theory of Computation},  pp. 100-105, 2003.  
	
	\bibitem{onlineIndependentSet} I. Caragiannisa, A. V. Fishkinb, C. Kaklamanisa, E. Papaioannoua. {\em Randomized Online Algorithms and Lower Bounds for Computing Large Independent Sets in Disk Graphs. Symposium on Mathematical Foundations of Computer Science},	155 (2) pp. 119–136, 2007.
	
	\bibitem{conjecture} R. Tiwari and M. T. Thai. {\em On Enhancing Fault Tolerance of Virtual Backbone in a Wireless Sensor Network with  Unidirectional Links Sensors. Theory, Algorithms, and Applications}, pp. 3-18, 2011.
	
	\bibitem{simulation} R. Tiwari, T. Mishra, Y. Li. {\em k-strongly Connected Dominating and Absorbing Set in Wireless Ad hoc Networks with Unidirectional Links. Wireless Algorithm, Systems and Applicaitons}, pp. 103-112, 2007.

\end{thebibliography}
\end{document}